\documentclass[12pt,letterpaper]{article}

\newif\ifisanonymous
\isanonymousfalse 

\newif\ifisTechnometrics
\isTechnometricsfalse 

\usepackage{amsmath,amssymb,amsfonts}
\usepackage{amsthm}

\usepackage{graphicx}
\usepackage[]{natbib}
\ifisTechnometrics
\else  
\fi
\usepackage{multirow}
\usepackage{authblk} 
\usepackage{booktabs}
\usepackage{caption}
\usepackage{bm}
\usepackage[english]{babel}
\usepackage[margin=1in]{geometry}
\usepackage[dvipsnames]{xcolor}
\usepackage{hyperref}
\hypersetup{
    colorlinks=true,
    citecolor=Blue
}

\newtheorem{theorem}{Theorem}
\newtheorem{proposition}{Proposition}

\newtheorem{corollary}[theorem]{Corollary}
\newtheorem{remark}{Remark}

\newcommand{\mat}[1]{\mathbf{#1}} 
\renewcommand{\vec}[1]{\boldsymbol{#1}}
\renewcommand{\b}{\vec{b}}
\newcommand{\y}{\vec{y}}
\newcommand{\yhat}{\hat \y}
\renewcommand{\vec}[1]{\bm{#1}}

\newcommand{\bhat}{\hat \b}
\newcommand{\hatb}{\bhat} 

\newcommand{\bhatd}{\bhat_d}
\newcommand{\bhatLS}{\hat \b_{\mathrm{LS}}}
\newcommand{\bhatfLiu}{\hat \b_{\mathrm{fLiu}}}
\newcommand{\bhatridge}{\hatb_{\mathrm{ridge}}}

\newcommand{\E}{\mathbb{E}} 
\newcommand{\R}{\mathbb{R}} 
\newcommand{\Q}{\mat{Q}_{\lambda,\alpha}}
\newcommand{\rank}{\operatorname{rank}}
\newcommand{\tr}{\operatorname{tr}}

\DeclareMathOperator{\MSE}{MSE}
\DeclareMathOperator{\Var}{Var}

\newcommand{\Zt}{\mat{Z}^{\top}}
\newcommand{\Zp}{\mat{Z}^{+}}
\newcommand{\Zmat}{\mat{Z}}
\newcommand{\Smat}{\mat{S}}
\newcommand{\Sp}{\Smat^{+}}
\newcommand{\Qmat}{\mat{Q}}

\newcommand{\SQ}{\Smat+\Qmat}
\newcommand{\SQinv}{(\Smat+\Qmat)^{-1}}
\newcommand{\Bmat}{\mat{B}}
\newcommand{\Ad}{\mat{A}_d} 
\newcommand{\Mmat}{\mat{M}}
\newcommand{\Rmat}{\mat{R}}
\newcommand{\Hd}{\mat{H}_{d}}
\newcommand{\In}{\mat{I}_{n\times n}}
\newcommand{\Ip}{\mat{I}_{m\times m}}
\newcommand{\Ipp}{\mat{I}_{m+1\times m+1}} 
\newcommand{\Wmat}{\mat{W}}

\newcommand{\fyd}{\hat{\vec{y}}_{d}}

\newcommand{\cardot}{\text{gen-ridge}} 

\ifisanonymous
\author{anonymous authors for peer review}
\else
\author[1]{Shaista Ashraf}
\author[2]{Stephen Becker}
\author[3]{Farrukh Javed}
\author[1]{Ismail Shah}

\affil[1]{Department of Statistics, Quaid-i-Azam University Islamabad, Pakistan}
\affil[2]{Department of Applied Mathematics, University of Colorado Boulder, USA}
\affil[3]{Department of Statistics, Lund University, Sweden}

\fi

\begin{document}

\ifisTechnometrics
\def\spacingset#1{\renewcommand{\baselinestretch}%
{#1}\small\normalsize} \spacingset{1}
\fi

\title{Functional Liu Regression for Scalar-on-Functional Models in High-Dimensional Settings}
\maketitle

\begin{abstract}
This study develops a functional Liu-type shrinkage estimator (\texttt{fLiu}) for scalar-on-function regression in the presence of strong multicollinearity and high-dimensional functional predictors. The approach extends the classical Liu estimator to the functional setting by combining directional shrinkage with smoothness regularization, providing flexible control over the bias--variance trade-off. Theoretical analysis is used to examine the behavior of the estimator and the associated parameter selection problem. In particular, an explicit mean squared error (MSE) decomposition is derived, characterizing the risk of the estimator in terms of variance reduction and shrinkage bias. This further yields an explicit optimal choice of the shrinkage parameter of the \texttt{fLiu} estimator through a one-dimensional convex risk minimization problem, leading to a practical plug-in tuning rule. Moreover, it is shown that in high-dimensional (underdetermined) settings, commonly used criterion such as GCV (and equivalently PRESS/LOO-CV) become constant with respect to the parameter d, thus uninformative for tuning. This provides a theoretical explanation for the predominant focus on the overdetermined regime in existing Liu-type methods.
 Numerical results demonstrate that the estimator achieves competitive predictive accuracy relative to existing methods. Implementation is carried out in \textsf{R} using the \texttt{fda} package, and in Python via the \texttt{fLiu.py} package developed for this study.
\end{abstract}

\ifisTechnometrics
\noindent%
\vfill
\newpage
\spacingset{1.8} 
\fi

\section{Introduction}
\label{Introduction}
Complex and high-dimensional data can challenge conventional statistical modeling techniques, particularly in settings where predictors are not scalar quantities but rather entire functions defined over a continuum such as time, space or frequency. Functional data analysis (FDA) provides a basic framework for analyzing such data, offering tools to model smooth and continuous functional predictors \citep{ramsay2005}. Functional regression models—central to FDA—are commonly classified into three types: functional-on-functional, functional-on-scalar, and scalar-on-functional regression. This study is conducted within the scalar-on-functional regression paradigm, which is especially relevant in applications such as climate science, biostatistics, and economics, and  assumes a scalar response which depends on (possibly multiple) functional predictors. 
Let $\mathcal{T}\subset \R$ represent the functional observation domain, and let $\mathcal{H}=L^2(\mathcal{T})$ denote the Hilbert space of square-integrable real-valued functions on $\mathcal{T}$, equipped with inner product and norm 
\begin{equation}
\langle f,g\rangle_{\mathcal H}
=
\int_{\mathcal T} f(t)g(t)\,dt, \quad \|f\|_{\mathcal H} = \sqrt{\langle f, f\rangle_{\mathcal{H}}}.
\end{equation}
For $i=1,\ldots,n$, the scalar-on-function regression model with $p$ functional predictors is
\begin{equation}\label{eq:FLRM}
y_i
=
\beta_0+\sum_{j=1}^p \langle X_{i,j},\beta_j\rangle_{\mathcal H}+\varepsilon_i,
\qquad
\E(\varepsilon_i)=0,\quad
\operatorname{Var}(\varepsilon_i)=\sigma^2
\end{equation}
where $X_{i,j}\in L^2(\mathcal{T})$ are the functional covariates. From here on, we formulate the theoretical development under the centered model, $\beta_0=0$, for simplicity of exposition.\footnote{In empirical applications, when the data are not centered, an intercept term can be included and estimated separately without affecting the main methodology.}
Each functional predictor is approximated using a truncated basis expansion with $K$ orthonormal basis functions $\{\phi_k\}_{k=1}^{K}$. 
Thus,
\[
X_{i,j}(s)\approx \sum_{k=1}^{K}\xi_{i,j,k}\phi_k(s),
\qquad i=1,\dots,n,\; j=1,\dots,p,
\]
where $\xi_{i,jk}$ is the basis score of the $j$th predictor for subject $i$. For $j=1,\ldots,p$, we represent the coefficient functions $\beta_j$ and their estimates $\widehat{\beta}_{j}$ as
\begin{equation}\label{eq:basis}
\beta_j(s)
=\sum_{k=1}^K b_{j,k}\phi_k(s),
\qquad
\widehat{\beta}_{j}(s)=\sum_{k=1}^K \hat b_{j,k}\phi_k(s).
\end{equation}
For each $j=1,\ldots,p$, define
\[
\b_j=(b_{j,1},\ldots,b_{j,K})^\top,
\qquad
\widehat{\b}_j=(\hat b_{j,1},\ldots,\hat b_{j,K})^\top.
\]
Stacking the true coefficient vectors across all predictors, define
\begin{equation}\label{coeffvec_b}
\b = (\b_1^\top,\ldots,\b_p^\top)^\top \in \mathbb{R}^m,
\qquad m = pK
\end{equation}
and likewise for $\hatb$. 
Then, by orthonormality of $\{\phi_k\}_{k=1}^K$,
\[
\|\widehat{\beta}_{j}-\beta_{j}\|_{L^2(\mathcal T)}^2
=
\|\widehat{\b}_j-\b_j\|_2^2.
\]
We define the design matrix $\Zmat\in\R^{n\times m}$ so that its $i$th row is determined so that 
\begin{align}
(\Zmat \b)_i = \sum_{j=1}^p \langle X_{i,j},\beta_j\rangle_{\mathcal H}
&= \sum_{j=1}^p \left\langle \sum_{k=1}^{K}\xi_{i,j,k}\phi_k,\sum_{k'=1}^K  b_{j,k'}\phi_{k'}\right\rangle_{\mathcal H} \notag \\
&= \sum_{j=1}^p \sum_{k'=1}^K b_{j,k'} \left( \sum_{k=1}^{K} \xi_{i,j,k}  \left\langle \phi_k, \phi_{k'}\right\rangle_{\mathcal H} \right) \label{eq:defZ_general}
\end{align}
and hence, again by orthonormality of $\{\phi_k\}_{k=1}^K$, its $i$th row is
\begin{equation} \label{eq:defZ}
\Zmat_i^\top=
(\xi_{i,(1,1)},\dots,\xi_{i,(1,K)},\xi_{i,(2,1)},\dots,\xi_{i,(p,K)}).
\end{equation}
Collecting the response entries $y_i$ into the vector $\y\in\R^n$ and likewise the errors $\varepsilon_i$ into $\boldsymbol{\varepsilon}\in\R^n$, then the finite-dimensional coefficient-space model for scalar-on-function regression is 
\begin{equation}\label{eq:Centred_model}
\y=\Zmat\b+\boldsymbol{\varepsilon},
\qquad
\E(\boldsymbol{\varepsilon})=0,
\qquad
\operatorname{Cov}(\boldsymbol{\varepsilon})=\sigma^2 \mat{I}_n.
\end{equation}

Despite its flexibility, scalar-on-function regression often faces challenges because the smooth structure of functional predictors induces substantial dependence among basis-expanded regressors. After basis expansion, these dependencies translate into multicollinearity in the coefficient space. Consequently, the design and Gram matrices may be ill-conditioned or singular, leading to unstable ordinary least squares estimates and inflated variance, particularly in high-dimensional settings where the number of basis coefficients approaches or exceeds the sample size. As a result, regularization plays a central role in ensuring stable estimation and reliable prediction in functional regression~\citep{tikhonov1977solutions,hoerl1970}. Functional ridge-type methods \citep{cardot2003,eilers1996} extend ridge regularization to FDA but apply uniform penalties across all covariates. More recent approaches consider structured regularization, variable selection, and clustering in functional regression models \citep{mehrotra2022}.

The Liu estimator offers a principled remedy by introducing controlled shrinkage that stabilizes estimation in such settings, while often retaining more signal in dominant directions than stronger ridge-type penalization. 
The Liu estimator \citep{liu1993}, further developed by \citet{liu2003}, \citet{kibria2003} and \citet{akdeniz2007}, works in tandem with regularization method like ridge regression by introducing a biasing parameter to improve estimation accuracy in linear models. 
Whereas ridge achieves stabilization through a single \(\ell_2\)-type penalty, the Liu approach combines shrinkage with a biasing transformation, which can reduce variance without requiring the same degree of overall shrinkage. A further advantage of Liu regularization is its explicit control of the bias-variance trade-off. As shown by \citet{liu1993}, for an appropriate choice of the Liu parameter, the Liu estimator can dominate ordinary least squares in terms of mean squared error (MSE), thereby offering a theoretically grounded and practically useful alternative to conventional ridge-based methods.
Later studies extended Liu-type estimation through ridge comparisons \citep{gruber2010}, two-parameter forms \citep{ozkale2007}, and robust variants \citep{filzmoser2016}. However, 
it has not been directly extended to functional regression.
Since Liu-type adaptations to functional regression remain largely unexplored,
this paper introduces a functional Liu estimator (\texttt{fLiu}) that integrates smoothness regularization with a Liu-type biasing parameter to address multicollinearity and instability in scalar-on-function regression. We generalize and strengthen Liu's original theorem by establishing mean squared error (MSE) improvement, and then study parameter selection. We select the tuning parameters $(\lambda,d,\alpha)$ using cross-validation criteria based on the linear smoother representation of the fitted values. In particular, we consider leave-one-out cross-validation (LOO-CV) using the predicted residual error sum of squares (PRESS) formula, as well as generalized cross-validation (GCV). For the overdetermined case, the method requires selecting three parameters, making direct grid search unattractive. We therefore develop a software package that uses continuous optimization and automatic differentiation to minimize the GCV or PRESS/LOO-CV criteria. For the underdetermined case, we uncover the phenomenon that the GCV and PRESS/LOO-CV objectives are insensitive to one of the parameters, and hence we propose an alternative plug-in estimator method to find the parameters.

Finally, we numerically compare the performance of \texttt{fLiu} with ordinary least squares (OLS), ridge regression, classical Liu estimation, and generalized ridge-regression 
functional regression methods, showing that the new method is practical and stable. 
The remainder of the paper is organized as follows. Section 2 introduces the methodological framework and presents the proposed functional Liu estimator. Section 3 studies its theoretical properties and risk behavior, while Section 4 presents an empirical application using the Canadian weather dataset. Section 5 concludes with a discussion of findings and potential directions for future research.

\section{Methodology: Functional Liu Shrinkage Strategy for \texorpdfstring{$p$}{p} Covariates}
\label{sec:Methodology}
Before introducing the new functional Liu-type estimator, we first review existing estimators and regularization techniques, all in the finite-dimensional context, i.e., for recovering the coefficients $\b\in\R^m$. We define the Gram matrix as
\begin{equation}\label{eq:Gram_matrix}
\Smat := \Zmat^\top \Zmat \in \R^{m\times m}.
\end{equation}

\subsection{Ordinary least squares estimator}
The ordinary least squares (OLS) estimator $\bhatLS$ minimizes the criteria $ \mathcal{J}(\b)=\|\y - \Zmat\b\|_2^2$, or equivalently, solves the normal equations
\begin{equation}\label{eq:normal}
\Smat \bhatLS = \Zmat^\top \y.
\end{equation}
The general solution to this is $\bhatLS = \Zp \y$ where $\Zp$ is the Moore-Penrose pseudo-inverse of $\Zmat$.
When $\Smat$ is nonsingular, this is equivalent to
\begin{equation}\label{eq:OLS_bhat}
\bhatLS = \Smat^{-1}\Zmat^\top \y.
\end{equation}
When $\Smat$ is singular, there are infinitely many solutions to the normal equations so among these we choose $\bhatLS = \Zp \y$ since it is the minimum norm solution.
Note that $\Zp = \Sp \Zmat^\top$ (see, e.g., \citep[Prop.~3.2]{barata2012moore}).

\subsection{Basis representation and roughness penalties}
\label{Basis and roughness penality}
To encourage smooth coefficient functions, it is common to modify $\mathcal{J}$ with a quadratic roughness penalty of the form
\[
\mathcal{J}_{\Rmat}(\b) = \mathcal{J}(\b) + \lambda \b^\top \Rmat \b
\]
where $\Rmat \in \R^{m\times m}$ is a symmetric positive semi-definite matrix $(\Rmat \succeq 0)$. The specific choice of $\Rmat$ depends on the basis system used to represent the functional coefficients since it is designed as finite-dimensional approximation of some function depending on the derivatives of the $\beta_j$. 
Below we consider two common basis choices: B-splines and the Fourier basis.
For simplicity, the derivations below assume $p=1$; generalizing to $p>1$ is straightforward other than dealing with indexing.

\subsubsection{Spline basis penalties}
Basis splines (or B-splines) are a basis widely used when coefficient functions exhibit local curvature, heterogeneous smoothness, or boundary effects. Splines provide substantial local flexibility through piecewise polynomial segments joined smoothly at knots. 
Note that since B-splines are not orthogonal,  $\Zmat$ cannot be defined using Eq.~\eqref{eq:defZ} but instead it should be defined using the more general Eq.~\eqref{eq:defZ_general}.

Smoothness is usually enforced through penalizing curvature, $\lambda\sum_{j=1}^p\|\beta_j''\|_\mathcal{H}^2$, and this is approximated as usual by using only the first $K$ basis elements. 
Specifically, if $\beta_j(s) =\sum_{k=1}^K  b_{j,k}\phi_{k}(s)$ then $\beta_j''(s) =\sum_{k=1}^K  b_{j,k}\phi_{k}''(s)$ so
\begin{equation}\label{eq:second_derivatives}
\|\beta_j''\|_\mathcal{H}^2 = \sum_{k=1}^K \sum_{k'=1}^K \langle b_{j,k} \phi_k'', b_{j,k'} \phi_{k'}''\rangle_\mathcal{H}.
\end{equation}
When $p=1$, this can be represented as 
\[
\|\beta_j''\|_\mathcal{H}^2 = \b^\top \Rmat \b,\quad \Rmat\in\R^{K\times K}, \quad \left[\Rmat\right]_{k,k'} = \langle \phi_k'', \phi_{k'}''\rangle_\mathcal{H}
\]
and $\Rmat$ is a Gram matrix and hence positive semi-definite.

Alternatively, since the splines are localized, one can make a crude approximation using finite-difference methods.
Let \(D^{(2)} \in \R^{(K-2)\times K}\) denote the second-order difference operator: for each coefficient vector \(\b_j\),
\[
(D^{(2)} \b_j)_k = b_{j,k}-2b_{j,k+1}+b_{j,k+2},
\qquad k=1,\dots,K-2.
\]
The associated roughness penalty is
\[
\|D^{(2)} \b_j\|_2^2
=
\b_j^\top \Rmat \b_j,
\qquad
\Rmat=D^{(2)\top} D^{(2)} \in \R^{K\times K}.
\]
Again, \(\Rmat \succeq 0\) by construction, though generally \(\Rmat \not\succ 0\).

\subsubsection{Fourier basis penalties}
The Fourier basis is particularly suitable for smooth periodic or seasonal processes, where the underlying signal exhibits approximately repeating cycles over the observation domain. Their orthogonality and efficient computation make them especially attractive. 
Fourier expansions are generally less suitable when strong local features or discontinuities are present since discontinuities would require a large $K$ in order to effectively capture the function.
The Fourier basis representation for the coefficient function \(\beta_j(t)\) with period \(T\)  (i.e., $\mathcal{T}=\R/(T\mathbb{Z})$) 
using $K=2\overline K + 1$ basis functions 
is
\[
\beta_j(t)
=
\frac{a_{0,j}}{2\pi}
+
\sum_{k=1}^{\overline K}
\left[
\frac{a_{k,j}}{\sqrt{\pi}}\sin(\omega_k t)
+
\frac{c_{k,j}}{\sqrt{\pi}}\cos(\omega_k t)
\right],
\qquad
\omega_k=\frac{2\pi k}{T}, \quad k=0,\dots,\overline K
\]
and, for $p=1$, $\b=\left[a_0,a_1,c_1,a_2,\ldots,c_{\overline{K}}\right]$.
Defining the basis functions to be 
$$\left\{\frac{1}{\sqrt{2\pi}}, \frac{1}{\sqrt{\pi}}\sin(\omega_1 \cdot), \frac{1}{\sqrt{\pi}}\cos(\omega_1 \cdot), \frac{1}{\sqrt{\pi}}\sin(\omega_2 \cdot), \ldots \right\}$$ and using $\phi_k'' = -\omega_k^2\phi_k$ and orthonormality, then (again for $p=1$) Eq.~\eqref{eq:second_derivatives} leads to 
\[
\Rmat=
\operatorname{diag}
\bigl(
0,\omega_1^2,\omega_1^2,\dots,\omega_K^2,\omega_K^2
\bigr)
\]
which is intuitive since it penalizes higher frequencies more.

\subsection{Generalized ridge estimator}
The classical ridge penalty is $\lambda\b^\top \b$, corresponding to $\b^\top\Qmat \b$ with $\Qmat=\lambda\Ip$. The penalty is rather generic but because $\Qmat \succ \mat{0}$ it is useful when $\Smat$ is singular or ill-conditioned.  In contrast, setting $\Qmat=\lambda\Rmat$, using the roughness matrix $\Rmat$ described in the previous section, is a natural regularizer in the FDA setting because this penalizes curvature of the function, thereby encouraging smoother solutions. However, typically $\Rmat$ is singular (i.e., anytime the basis includes a constant function) so it may not fully ameliorate singularity and ill-conditioning in $\Smat$.

In order to combine the benefits of both $\Qmat=\Ip$ and $\Qmat=\Rmat$, 
we therefore define the hybrid penalty matrix
\begin{equation}\label{eq:Q_hybrid}
\Q
=
\lambda\bigl(\alpha \Ip+(1-\alpha)\Rmat\bigr),
\qquad
\lambda>0,\quad \alpha\in[0,1].
\end{equation}
Here, $\lambda$ controls the overall magnitude of penalization, while $\alpha$ determines the balance between isotropic ridge shrinkage and basis-dependent roughness regularization. We assume, without loss of much generality, that $\Rmat$ is scaled so that
$
\Rmat \preceq \Ip$.
Specifically, we define the functional generalized ridge estimator as
\[
\bhat_\text{gen-ridge}
=
(\Smat+\Q)^{-1}\Zmat^\top \y.
\]
We note that an estimator of a similar form was analyzed by \citet{cardot2003} and \citet{ramsay2002}.
\begin{remark}\label{rmk:Q_posdef}
Assuming \(\Rmat \succeq 0\) and it is scaled such that $\|\Rmat\|\le 1$, then $\Q\succ \mat{0}$ is guaranteed if $\lambda >0$ and either 
\begin{enumerate}
    \item $\alpha>0$, or
    \item $\Rmat\succ \mat{0}$ (and any $\alpha\in[0,1]$).
\end{enumerate}
\end{remark}
When $\lambda$ and $\alpha$ are fixed and clear from context, we write just $\Qmat$ instead of $\Q$ for conciseness.

\subsection{Proposed Functional Liu-type estimator} 
The generalized ridge estimator with penalty matrix $\Qmat$ minimizes $\mathcal{J}(\b)=\|\y-\Zmat\b\|_2^2 + \b^\top \Qmat\b$. To introduce Liu-type biasing, we 
define the estimator $\bhatfLiu$ which minimizes the following centered penalized criterion
\begin{equation}\label{eq:fLiu_objective}
\mathcal{J}_{\lambda,d,\alpha}(\b)
:=
\|\y - \Zmat\b\|_2^2
+
\bigl(\b - d\,\bhatLS\bigr)^\top
\Qmat
\bigl(\b - d\,\bhatLS\bigr),
\qquad d\in\R,
\end{equation}
where $d$ is the Liu-type biasing parameter. The term $(\b-d\bhatLS)$ anchors the penalized solution toward (if $d>0$) or away from (if $d<0$) the least-squares direction in the $\Qmat$-geometry, which can mitigate over-shrinkage in weakly identified directions. Differentiating Eq.~\eqref{eq:fLiu_objective} with respect to $\b$ and setting the gradient to zero gives
\begin{equation}\label{eq:fLiu_normal}
(\Zmat^\top \Zmat + \Qmat)\,\bhatfLiu 
=
\Zmat^\top \y + d\Qmat\,\bhatLS.
\end{equation}
If $\Smat$ is nonsingular or $\Qmat$ is nonsingular (e.g., the conditions of Remark~\ref{rmk:Q_posdef} hold), then 
 $\Zmat^\top \Zmat + \Qmat = \SQ$ is invertible, and the proposed estimator is
\begin{equation}\label{eq:fLiu_estimator}
\bhatfLiu = \bhatfLiu(\lambda,d,\alpha)
=
\SQinv\bigl(\Zmat^\top \y + d \Qmat\bhatLS\bigr),
\qquad d\in\R.
\end{equation}
The estimator is well-defined for any $d\in\R$, though it is traditional to restrict it to $d\in[0,1]$~\cite{liu1993,liu2003}.

When $d=0$, Eq.~\eqref{eq:fLiu_estimator} reduces to the usual hybrid ridge estimator under $\Qmat$. When $d=1$ and $\Smat \succ 0$, 
the estimator reduces exactly to OLS regardless of the values of $\lambda$ and $\alpha$:
\[
\bhatfLiu(\lambda,1,\alpha)=\bhatLS.
\]
This ``$d=1$ returns OLS'' property is central to the risk comparison in Section~\ref{Risk Property}; the proof of this will be apparent (see Remark~\ref{rmk:d1}) once we introduce  Eq.~\eqref{eq:fLiu_as_linear_of_ols} in later sections.

\subsection{Role of Tuning Parameters}
The proposed functional Liu estimator depends on three tuning parameters $(\lambda, d, \alpha)$ each playing a fundamentally distinct role in the geometry and magnitude of shrinkage. Although they interact in the final estimator, their statistical functions are conceptually different.

\paragraph{Role of $\lambda$ (Overall regularization strength):}
The parameter $\lambda>0$ controls the overall magnitude of regularization through the hybrid penalty in Eq.~\eqref{eq:Q_hybrid}.
Increasing $\lambda$ uniformly scales the penalty matrix $\Q$ and therefore increases shrinkage intensity in all directions of the coefficient space. From the estimator equation~\eqref{eq:fLiu_estimator}, larger $\lambda$ inflates $\Q$, enlarging $\Smat + \Q$ and thus shrinking the estimator toward $\bhatLS$. When $\lambda=0$, the estimator reduces to OLS (provided $\Smat$ is invertible). As $\lambda\to\infty$, shrinkage dominates the data-driven component and the variance will go to zero (but with high bias).

\paragraph{Role of $\alpha$ (Geometry of shrinkage):}
The parameter $\alpha \in [0,1]$ determines the structural composition of the hybrid penalty in Eq.~\eqref{eq:Q_hybrid}. It balances two distinct shrinkage mechanisms:
\begin{itemize}
	\item $\alpha = 1$: $\Q = \lambda \Ip$ produces classical ridge shrinkage, which penalizes all coefficient directions equally and does not explicitly enforce smoothness but is very effective at counteracting ill-conditioning of $\Smat$.
	\item $\alpha = 0$: $\Q = \lambda \mat{R}$ yields pure roughness regularization, penalizing curvature of the functional coefficients but may not be positive definite.
	\item $0 < \alpha < 1$: yields a hybrid structure that combines global magnitude control with smoothness enforcement.
\end{itemize}
Geometrically, $\alpha$ determines the shape of shrinkage in coefficient space, whereas $\lambda$ determines its magnitude.

\paragraph{Role of $d$ (Liu-type shrinkage adjustment):}
The parameter $d \in [0,1]$ is the Liu-type biasing parameter and enters the estimator through
Eq.~\eqref{eq:fLiu_estimator}.
Unlike $\lambda$ and $\alpha$, which define the penalty matrix $\Q$, the parameter $d$ modifies how strongly the penalized estimator is anchored toward the least-squares direction---it translates the shrinkage. Two important boundary cases illustrate its role:
\begin{itemize}
	\item $d = 0$: the estimator reduces to the standard hybrid ridge/Tikhonov estimator
	\[
	\bhatridge = (\Smat+\Q)^{-1} \Zmat^\top \y.
	\]
	\item $d = 1$: if $\Smat$ is invertible, the estimator reduces exactly to OLS:
	\[
	\bhatfLiu(\lambda,1) = \bhatLS.
	\]
\end{itemize}
The parameter $d$ governs the bias-variance trade-off along a continuous path connecting OLS ($d=1$) and penalized regression ($d=0$). As $d$ decreases from $1$ toward $0$, the estimator introduces bias but reduces variance through shrinkage. 

To summarize,
\[
\begin{aligned}
\lambda &:\ \text{magnitude of regularization}, \\
\alpha  &:\ \text{geometric structure of the penalty}, \\
d       &:\ \text{bias--variance tuning relative to OLS}.
\end{aligned}
\]
The separation of these roles allows the proposed estimator to decouple 
(i) stability and smoothness enforcement ($\lambda,\alpha$) from 
(ii) Liu-type risk adjustment ($d$), 
providing additional flexibility beyond classical ridge or roughness penalization.	
    
\subsection{Parameter selection}

We select the tuning parameters \((\lambda,d,\alpha)\) using cross-validation criteria based on the linear smoother representation of the fitted values \citep{craven1979}. For the \texttt{fLiu} estimator, the fitted response can be written as
\[
\yhat=\Zmat\bhatfLiu(\lambda,d,\alpha)=\Hd\y,
\]
where the smoother matrix is
\begin{align}
\Hd
&=
\Zmat\SQinv\Zmat^\top
+d\,\Zmat\SQinv\Qmat\,\Smat^{-}\Zmat^\top \notag\\
&=
\Zmat\SQinv(\Ip+d\Qmat\Smat^{-})\Zmat^\top,
\label{eq:smoother_matrix}
\end{align}

For linear smoothers, the leave-one-out cross-validation (LOO-CV) criterion admits the closed-form expression
\[
\mathrm{CV}(\lambda,d,\alpha)
=
\frac{1}{n}\sum_{i=1}^n
\left(
\frac{y_i-\hat y_{d,i}}{1-H_{d,ii}}
\right)^2,
\]
where \(\hat y_{d,i}\) is the \(i\)-th fitted value and \(H_{d,ii}\) denotes the \((i,i)\)-th diagonal entry of \(\Hd\);
when the $\frac{1}{n}$ scaling is absent, this is also known as the predicted residual error sum of squares (PRESS) criterion.

A widely used approximation to LOO-CV is generalized cross-validation (GCV), which replaces the individual leverage values \(H_{d,ii}\) by their average \(\tr(\Hd)/n\), and thus yields
\[
\mathrm{GCV}(\lambda,d,\alpha)
=
\frac{\|\y-\yhat\|_2^2/n}{\left(1-\tr(\Hd)/n\right)^2}.
\]
The GCV criterion is 
especially convenient when the diagonal entries of \(\Hd\) are costly or unstable to evaluate individually. In our implementation, we evaluate LOO-CV/PRESS  and GCV through the common smoother matrix representation in Eq.~\eqref{eq:smoother_matrix}, and select \((\lambda,d,\alpha)\) by minimizing the chosen score. 

Minimizing the score function for a single parameter can be done via grid search: for example, choose 100 values of $\lambda$ and then just evaluate the score at each value. However, because we search for three parameters, grid search is less appealing. If each parameter were quantized into 100 values, this would require evaluating the score $100^3 = 10^6$ times and hence take a long time, since each evaluation requires solving a linear system.  

To circumvent this issue, we instead treat this as a continuous optimization problem over $\R^3$ and have developed a publicly available Python package \texttt{fLiu.py}\footnote{
\ifisanonymous
link redacted in anonymous version
\else
available at \href{https://github.com/stephenbeckr/functional-Liu}{github.com/stephenbeckr/functional-Liu}.
\fi}.
The package uses the Sequential Least Squares Programming (SLSQP) method from the \texttt{scipy.optimize} package, which requires calculation of gradients of the score function with respect to the parameters. Our package does these gradient calculations automatically using the automatic differentiation capabilities of the \texttt{JAX} package. 
To mitigate the issue of the optimizer being stuck at local (but non-global) minimizers, we first perform a coarse grid search with just 5 grid points per dimension and use the outcome of this grid search as the initialization for the SLSQP method. 
Typically the optimal parameters are found to very high accuracy in at most a few dozen function and gradient evaluations, which is a large computational savings over the millions of function evaluations required by a fine grid search.

\section{Theoretical Properties}
This section examines the statistical properties of the proposed fLiu estimator. We first derive its sample moments, including bias and variance, and then study its risk behavior to assess efficiency and stability relative to competing estimators.

\subsection{Moments of the estimator}

Under the model in Eq.~\eqref{eq:FLRM} and the assumptions in Eq.~\eqref{eq:Centred_model}, 
then it is well-known that 
the ordinary least squares estimator $\bhatLS=\Zp\y = \Sp\Zmat^\top \y$ has mean, bias, variance and mean squared error given by
\begin{align}
\E[\bhatLS] &= \Zp\Zmat\b, &  \E[\bhatLS] &= \b\quad\text{if }\Smat\text{ invertible}\label{OLS_Mean} \\
\mathrm{Bias}[\bhatLS] &= (\Zp\Zmat-\Ip)\b  &\mathrm{Bias}[\bhatLS]&=\vec{0}\quad\text{if }\Smat\text{ invertible}, \label{OLS_Bias} \\
\mathrm{Var}[\bhatLS] &= \sigma^2 \Sp & \mathrm{Var}[\bhatLS] &=\sigma^2\Smat^{-1}\quad\text{if }\Smat\text{ invertible}, \label{OLS_Var} \\
\mathrm{MSE}[\bhatLS] &= \sigma^2 \tr(\Sp)+\left\| \Zp\Zmat\b-\b \right\|_2^2 & 
\mathrm{MSE}[\bhatLS] &=\sigma^2 \tr(\Smat^{-1})\quad\text{if }\Smat\text{ invertible} \label{OLS_MSE}
\end{align}
using the easily-derived identity $\MSE[\bhat] = \tr( \Var[\bhat] ) + \left\| \mathrm{Bias}[\bhat] \right\|_2^2$ and that $\Sp \Smat\Sp = \Sp$. Specifically, the variance identity follows because $\mathrm{Var}[\bhatLS]=\mathrm{Var}[\Sp\Zmat^\top\y] = \Sp\Zmat^\top \mathrm{Var}[\y] \Zmat\Sp = \sigma \Sp\Smat\Sp$.

We will now derive the moments for the functional Liu-type estimator. 
As before, we will often write $\mat{Q}$ instead of $\Q$ (and $\bhatfLiu$ instead of $\bhatfLiu(\lambda,d,\alpha)$) when $\lambda$ and $\alpha$ are fixed and clear from the context.
Using the OLS normal equation $\Zmat^\top \y=\Smat\bhatLS$, we can rewrite Eq.~\eqref{eq:fLiu_estimator} as
\begin{align}
\bhatfLiu(\lambda,d,\alpha) &= \SQinv\bigl(\Zmat^\top \y + d\Qmat\bhatLS\bigr) \\
&=\SQinv\bigl(\Smat\bhatLS + d \Qmat\bhatLS\bigr) \notag\\
&=\underbrace{\SQinv(\Smat+d\Qmat)}_{\Ad}\bhatLS. \label{eq:fLiu_as_linear_of_ols}
\end{align}
so that we can succinctly write 
\begin{equation}\label{eq:fLiu_Ahatb}
\bhatfLiu(\lambda,d,\alpha)=\Ad \bhatLS.
\end{equation}
\begin{remark}[Reduction to OLS for $d=1$]
\label{rmk:d1}
Because $\Ad=\Ip$ when $d=1$, we see from Eq.~\eqref{eq:fLiu_as_linear_of_ols} that for any $\alpha$ and $\lambda$, $\bhatfLiu(\lambda,d,\alpha)=\bhatLS$ if $d=1$.
\end{remark}

\begin{proposition}\label{prop:moments}
For any real parameters $\lambda,d,\alpha$, assuming $\SQ$ is invertible for any $\Qmat\succeq\mat{0}$
 (e.g., if $\Qmat=\Q$, this is guaranteed if either $\mat{Z}$ has full column rank, or 
if any of the conditions of Remark~\ref{rmk:Q_posdef} hold)
then $\bhatfLiu=\bhatfLiu(\lambda,d,\alpha)$ is well-defined and 
\begin{align}
\E\left[\bhatfLiu\right]&=
\SQinv(\Smat+d\Qmat)\b .   \label{eq:fLiu_mean} \\
\mathrm{Bias}\left[\bhatfLiu\right]
&=
(d-1)\SQinv\Qmat\b \label{eq:fLiu_bias} \\
\mathrm{Var}\left[\bhatfLiu\right]
&=
\sigma^2\SQinv(\Smat+d\Qmat)\Smat^{+}(\Smat+d\Qmat)\SQinv 
= \sigma^2 \Ad  \Smat^{+} \Ad^\top \label{eq:fLiu_var}
\end{align} 
\end{proposition}
\begin{proof}
The statement for the mean, Eq.~\eqref{eq:fLiu_mean}, follows immediately because $\E[ \mat{A}_d\bhatLS  ]=\mat{A}_d \E[ \bhatLS] $ and $\E[\bhatLS] = \b$.

To calculate the bias,
\begin{align*}
\mathrm{Bias}\left[\bhatfLiu(\lambda,d,\alpha)\right]
&:=\E[\bhatfLiu(\lambda,d,\alpha)]-\b \\
&=\Bigl(\SQinv(\Smat+d\Qmat)-\Ip\Bigr)\b \\
&=\Bigl(\SQinv(\Smat+d\Qmat)-\SQinv(\SQ)\Bigr)\b\\
&= (d-1)\SQinv\Qmat\b.
\end{align*}

Finally, using $\mathrm{Var}[\bhatLS] = \sigma^2 \Smat^{+}$ from Eq.~\eqref{OLS_Var}, 
the variance is
\begin{align*}
\mathrm{Var}[\bhatfLiu] &= 
\Ad \mathrm{Var}[\bhatLS]\Ad^\top \\
&= \Ad (\sigma^2 \Smat^{+} )\Ad^\top \\
&= \sigma^2\SQinv(\Smat+d\Qmat)\Smat^{+}(\Smat+d\Qmat)\SQinv.
\end{align*}
\end{proof}

\subsection{Risk Properties}
\label{Risk Property}

A celebrated result of the original Liu estimator is the fact that there is some $d\in(0,1)$ for which it improves on the least squares estimator in terms of the MSE.  Below, we generalize this theorem considerably and show that a similar result holds for the functional Liu estimator.

\begin{theorem}[Risk improvement of the Functional Liu estimator]\label{thm:risk_fliu}

Consider the centered finite-dimensional model of Eq.~\eqref{eq:Centred_model}. 
Let $\Qmat \succ \mat{0}$ be any symmetric positive definite matrix (e.g., $\Qmat=\Q$ and either condition of Remark~\ref{rmk:Q_posdef} holds). Since $\Smat=\Zmat^\top \Zmat \succeq \mat{0}$, it follows that $\Smat+\Qmat \succ \mat{0}$, and therefore $\SQinv$ exists.
Let $\Qmat=\lambda(\alpha \Ip+(1-\alpha)\mat{R})$ with $\lambda>0$ and $\alpha\in(0,1]$, and define
\[
\bhatLS= \Smat^{+}\Zmat^\top \y,
\qquad
\bhatd=\SQinv\bigl(\Zmat^\top \y+d\Qmat\bhatLS\bigr).
\]

Then $\exists d\in[0,1)$ such that 
\[
\MSE[\bhatd] :=\E\|\bhatd-\b\|_2^2 < \E\|\bhatLS-\b\|_2^2 = \MSE[\bhatLS]
\]
where the expectation is over the noise $\vec{\varepsilon}$ in Eq.~\eqref{eq:Centred_model}. The value of $d$ may depend on $\Zmat,\b$ and $\Qmat$.
\end{theorem}

\begin{proof}
Considering $\MSE[\bhatd]=:g(d)$ as a function of $d$, the proof strategy is to show that the derivative, $g'(d)$, exists and is positive at $d=1$, hence it directly follows from the definition of the derivative that for some $\tilde{d}<1$ sufficiently close to $1$ that $g(\tilde{d})<g(1)$, and further that $\MSE[\bhatLS]=g(1)$ as seen in Remark \ref{rmk:d1}, and hence the result follows.

Using the results of Prop.~\ref{prop:moments} for the bias and variance, we calculate: 
\begin{align*}
\MSE[\bhatd] &:= \E \|\bhatd - \b\|_2^2 \\
&= \tr( \Var[\bhatd] ) + \big\| \mathrm{Bias}[\bhatd] \big\|_2^2 \\
&=\sigma^2\tr\left( \SQinv(\Smat+d\Qmat) \Smat^{+} (\Smat+d\Qmat)\SQinv\right) + (d-1)^2\|\SQinv\Qmat\b\|_2^2 \\
&=\sigma^2\tr\left( \SQinv(\Smat+2d\Qmat + d^2\Qmat\Smat^{+} \Qmat )\SQinv\right) + (d-1)^2\|\SQinv\Qmat\b\|_2^2 \\
&=\underbrace{\sigma^2\tr\left( \SQinv\Smat\SQinv\right)}_{c_0}
+d\underbrace{2\sigma^2\tr\left( \SQinv\Qmat\SQinv \right)}_{c_1} \\
&\quad\ldots +d^2\underbrace{\sigma^2\tr\left( \SQinv \Qmat\Smat^{+}\Qmat \SQinv \right)}_{c_2} + (d-1)^2\underbrace{\|\SQinv\Qmat\b\|_2^2}_{c_3} \\
&= c_0 + c_1 d + c_2 d^2 + c_3 (d-1)^2
\end{align*}
so we see that $g$ is a quadratic polynomial, hence differentiable, with derivative 
\[
g'(d) = c_1 + 2d c_2 + 2 c_3(d-1)
\]
so $g'(1) = c_1+2c_2$. Since $\Qmat\succ \mat{0}$ and $\Smat\succeq \mat{0}$ then $\SQinv \Qmat\Smat^{+}\Qmat \SQinv\succeq 0$, so $c_2$ is the product of $d^2$ and the trace of a positive semidefinite matrix, hence $c_2\ge 0$, and likewise $c_1>0$ for similar reasons, thus $g'(1)>0$.
\end{proof}

From the proof of the theorem is the following immediate corollary:
\begin{corollary}[MSE is a quadratic] \label{cor:quadratic}
Assume the conditions of Theorem~\ref{thm:risk_fliu} hold. For fixed $S$, $Q$, $b$, and $\sigma^2$, define 
$$ g(d) =     \MSE[\bhatd] := \E \|\bhatd - \b\|_2^2, \quad d\in\R.$$
Then $g$ is a quadratic function and can be written as 
\[
g(d)=(c_2+c_3)d^2+(c_1-2c_3)d+(c_0+c_3),
\]
where
\begin{align*}
c_0&=\sigma^2 \operatorname{tr}\left(\SQinv \Smat\SQinv \right), \\
c_1&=2\sigma^2\tr\left( \SQinv\Qmat\SQinv \right), \\
c_2&=\sigma^2\tr\left( \SQinv \Qmat\Smat^{+}\Qmat \SQinv \right), \\
c_3&=\|\SQinv\Qmat\b\|_2^2.
\end{align*}
\end{corollary}

Theorem \ref{thm:risk_fliu} shows that the mean squared error depends explicitly on the Liu parameter $d$, suggesting that appropriate tuning of $d$ is important. We next characterize the optimal choice of $d$ under the derived risk expression.

\begin{proposition}[Convexity of the MSE in $d$ and the optimal Liu parameter]
\label{prop:d-opt_convexity}
Assume the conditions of Theorem~\ref{thm:risk_fliu} hold
and define $g$ as in Corollary~\ref{cor:quadratic}. 
If $\Qmat\neq \mat{0}$ and $\b\neq \vec{0}$, then $g(d)$ is strongly convex and therefore admits a unique global minimizer. The unconstrained minimizer is
\[
d_{\mathrm{opt}}=\frac{2c_3-c_1}{2(c_2+c_3)}
\]
and $d_\text{opt} \le 1$, and furthermore $d_\text{opt} < 1$ if $\Qmat \succ \mat{0}$.
\end{proposition}

\begin{proof}
Taking the derivative of $\MSE[\bhatd]$ yields 
\[
g'(d)=2(c_2+c_3)d+(c_1-2c_3),
\qquad
g''(d)=2(c_2+c_3).
\]
Thus, since $\Qmat\neq \mat{0}$ and $\b\neq \vec{0}$, it follows that $c_3 >0$ so $c_2+c_3>0$ and thus $g$ is strongly convex and has a unique global minimizer. 
Solving for $d$ to satisfy $g'(d)=0$ 
yields
\begin{equation}
\label{plug-in estimator-d}
d_{\mathrm{opt}}=\frac{2c_3-c_1}{2(c_2+c_3)}.   
\end{equation}
Moreover,
\begin{align*}
d_{\mathrm{opt}} = \frac{2c_3-c_1}{2(c_2+c_3)} 
&= \frac{c_3}{c_2+c_3} - \frac{c_1}{2(c_2+c_3)} \\
&\le \frac{c_3}{c_2+c_3} \\
&\le 1
\end{align*}
since \(c_2\ge 0\), \(c_3\ge 0\), and \(c_2+c_3>0\).
Whenever $c_1>0$ or $c_2 > 0$, we see that $d<1$.
The condition $\Qmat \succ \mat{0}$ guarantees $c_1>0$ and $c_2>0$.
\end{proof}

\subsubsection{Plug-in tuning rule for \texorpdfstring{$d$}{d}}
Proposition~\ref{prop:d-opt_convexity} gives the risk-optimal Liu parameter in closed form when the unknown quantities are available. In practice, the coefficients $c_0,c_1,c_2$ and $c_3$ are unknown because they depend on the unknown quantities $\b$ and $\sigma^2$. One approach is to estimate these coefficients by plugging in estimates for $\b$ and $\sigma^2$.

For fixed $(\lambda,\alpha)$, let
\[
\Qmat = \Q
\qquad
\bhat_\cardot=\SQinv \Zmat^\top \y,
\]
where $\hatb_{\mathrm{gen-ridge}}$ is the generalized ridge estimator. Let $\hat\sigma^2$ be an estimate of the error variance. We then define
\begin{align*}
c_1&=2\hat\sigma^2\tr\left( \SQinv\Qmat\SQinv \right), \\
c_2&=\hat\sigma^2\tr\left( \SQinv \Qmat\Smat^{+}\Qmat \SQinv \right), \\
c_3&=\left\|\SQinv\Qmat\bhat_\cardot\right\|_2^2.
\end{align*}

The plug-in estimator of the optimal Liu parameter is obtained by replacing $c_1,c_2,c_3$ in \eqref{plug-in estimator-d} by their empirical counterparts. Thus,
\[
\hat d_{\mathrm{plug}}
=
\frac{2\hat c_3-\hat c_1}{2(\hat c_2+\hat c_3)}.
\]
This estimator provides a data-driven choice of the Liu shrinkage parameter for fixed \((\lambda,\alpha)\). As an additional practical alternative, one may enforce the classical Liu parameter restriction \(d\in[0,1]\) by projecting the unconstrained plug-in estimate onto the admissible interval; $\hat d_{\mathrm{plug}} \le 1$ happens by construction, but one may project it to be nonnegative if so desired, giving 
\[
\hat d_{\mathrm{proj}}
=
\max\left\{0,\hat d_{\mathrm{plug}}\right\}.
\]

\subsubsection{Under-determined case}
Our next theorem concerns choosing parameter fitting via the GCV or leave-one-out cross-validation (LOO-CV) techniques. While the functional Liu estimator is still well-defined in the under-determined case, we are able to show in Theorem~\ref{thm:degeneracy of gcv} that the GCV and LOO-CV criteria are constants with respect to $d$, hence they are not helpful criteria for choosing an appropriate $d$ value (and so $\hat d_{\mathrm{plug}}$ may be particularly helpful). This may partially explain why most prior works on the Liu estimator only considered the over-determined case.
\begin{theorem}[Degeneracy of GCV with respect to $d$]
\label{thm:degeneracy of gcv}
Let $\Zmat \in \R^{n\times m}$ have full row rank (i.e., $\rank(\Zmat)=n$, so necessarily $n\le m$) i.e., the underdetermined (high-dimensional) regime, and define $\Smat=\Zt\Zmat$ and $\bhatLS=\Zp \y$ where $\Zp$ is the Moore--Penrose pseudo-inverse.
Let $\Qmat \in \R^{m\times m}$ be fixed such that $\SQ$ is invertible. For $d\in\R$, define the functional Liu estimator 
\[
\bhatd=\SQinv\bigl(\Zt \y + d\,\Qmat \bhatLS\bigr),
\qquad
\fyd=\Zmat\bhatd.
\]
Then 
\[
\fyd=\Hd \y,
\qquad
\Hd=\In-(1-d)\Bmat,
\]
where the matrix 
\(
\Bmat=\Zmat\SQinv\Qmat\Zp,
\)
is independent of $d$. Consequently, for every $d\neq 1$, 
\[
\mathrm{GCV}(d)
:=
\frac{\|\y-\fyd\|^2}{\bigl(n-\tr(\Hd)\bigr)^2}
= \frac{\|\Bmat\y\|^2}{\tr(\Bmat)^2}
\]
is independent of $d$, and the PRESS/LOO-CV criterion
\[
\mathrm{PRESS}(d)
:=
\sum_{i=1}^{n}
\left(
\frac{y_i-\hat y_{d,i}}{1-H_{d,ii}}
\right)^2
\]
is also independent of $d$ whenever defined, where $H_{d,ii}$ denotes the $(i,i)$ entry of $\Hd$.
\end{theorem}

\begin{proof}
Since $\bhatLS=\Zp\y$ is a solution to the normal equations, it means that $\Smat\bhatLS=\Zt\y$, so 
using this in the definition of $\bhatd$,
\begin{align*}
\bhatd
&=\SQinv\bigl(\Zt \y + d\,\Qmat \bhatLS\bigr) \\
&=\SQinv(\Smat\bhatLS + d\Qmat\bhatLS) \\
&=\SQinv(\Smat+d\Qmat)\bhatLS \\
&=\SQinv(\Smat+d\Qmat)\Zp\y.
\end{align*}
Hence the fitted values are
\begin{align*}
\fyd=\Zmat\bhatLS &=\Zmat\SQinv(\Smat+d\Qmat)\Zp\y \\
&= \Zmat\SQinv\bigl(\SQ-(1-d)\Qmat\bigr)\Zp\y \\
&= \Zmat\SQinv(\SQ)\Zp\y
-(1-d)\underbrace{\Zmat\SQinv\Qmat\Zp}_\Bmat\y \\
&= \Zmat\Zp\y - (1-d)\Bmat \y.
\end{align*}
By the assumption that $\Zmat$ has full row rank, then $\Zp=\Zt(\Zmat\Zt)^{-1}$, so $\Zmat\Zp=\In$, hence we have
\begin{equation}\label{eq:fyd}
\fyd=\Hd\y = \y - (1-d)\Bmat\y,\quad\text{where}\quad \Hd=\In-(1-d)\Bmat.
\end{equation}

Calculating the GCV and using Eq.~\eqref{eq:fyd} gives
\begin{align*}
\mathrm{GCV}(d) &= \frac{\|\y-\fyd\|^2}{\bigl(n-\tr(\Hd)\bigr)^2} \\
&= \frac{\| (1-d)\Bmat\y \|^2}{\bigl(n-\tr(\In-(1-d)\Bmat)\bigr)^2} \\
&= \frac{(1-d)^2\|\Bmat\y\|^2}{(1-d)^2\tr(\Bmat)^2}
\end{align*}
which, for $d\neq 1$, gives the desired result.

Calculating the LOO-CV error using the PRESS formula gives
\begin{align*}
    \mathrm{PRESS}(d)
&=
\sum_{i=1}^{n}
\left( \frac{y_i-\hat y_{d,i}}{1-H_{d,ii}}
\right)^2 \\
&=
\sum_{i=1}^{n}
\left( \frac{y_i-(y_i-(1-d)(\Bmat\y)_i)}{1-(1-(1-d)B_{ii})}
\right)^2 \\
&=
\sum_{i=1}^{n}
\frac{(1-d)^2(\Bmat\y)_i^2}{(1-d)^2B_{ii}^2}\\
&= \sum_{i=1}^{n}
\frac{(\Bmat\y)_i^2}{B_{ii}^2}
\end{align*}
whenever $d\neq 1$, where $B_{ii}$ denotes the $(i,i)$ element of the matrix $\Bmat$. Hence, whenever this is defined (i.e., when all $B_{ii}\neq 0$), it is independent of $d$.
\end{proof}

The following corollary is immediate:
\begin{corollary}[Special Case: Classical Liu estimator]
If $\Qmat=\Ip$, then
\[
\bhatd=\SQinv\bigl(\Zt\y+d\bhatLS\bigr),
\]
and
\[
\Hd=\In-(1-d)\Mmat,
\qquad
\Mmat=\Zmat(\Smat+\Ip)^{-1}\Zp.
\]
Thus, for $0\le d<1$, both GCV and PRESS are independent of $d$.
\end{corollary}

\begin{remark}
The plug-in rule is particularly useful in situations where generalized cross-validation (GCV) or leave-one-out cross-validation (LOO-CV) do not vary with respect to $d$. In such cases, these criteria cannot guide the choice of the Liu parameter. The proposed plug-in expression provides a simple and interpretable way to select $d$ based on the estimated bias--variance trade-off of the estimator. In particular, Proposition~\ref{prop:d-opt_convexity} shows that, for fixed $(\lambda,\alpha)$, the selection of $d$ can be formulated as a one-dimensional risk minimization problem with an explicit solution. 
\end{remark}

\section{Empirical Analysis: Canadian Weather Data}
This empirical study has three primary objectives: (i) to evaluate the performance of the proposed functional Liu estimator (\texttt{fLiu}) in terms of estimation stability and predictive accuracy, and to compare it with existing methods such as OLS, Ridge, Liu, and the generalized ridge regression estimator (with roughness)  under a GCV-based tuning framework; (ii) to predict total annual precipitation at Canadian weather stations using the pattern of temperature variation over the year; and (iii) to investigate whether variations in climatic conditions across stations can be explained by the seasonal structure of temperature and precipitation when modeled as correlated functional predictors. 

The application uses the well-known Canadian weather data set, available in the \texttt{fda} package in \textsf{R}, originally used in  \citet{ramsay2002} and later used in \citet{RamsayHookerGraves2009}. In particular, the accompanying Canadian weather demonstration code by Ramsay\footnote{available at \href{https://www.psych.mcgill.ca/misc/fda/downloads/FDAfuns/R/demo/canadian-weather.R}{www.psych.mcgill.ca/misc/fda/downloads/FDAfuns/R/demo/canadian-weather.R}} served as an initial reference framework for the regularized functional ridge implementation later extended here to the proposed functional Liu estimator. All empirical analysis reported in this paper, including the Canadian weather application, were carried out using our \texttt{fliu.py} package in \texttt{Python} together with the \texttt{fda} package in \texttt{R} software.

\subsection{Data Description and Model Framework}

We illustrate the proposed \texttt{fLiu} procedure using the Canadian weather dataset. This dataset contains averaged annual cycles of monthly mean temperature and monthly precipitation for $35$ Canadian weather stations, obtained by averaging observations over the period 1960-1994. Consequently, each station is represented by two smooth climatic curves corresponding to temperature and precipitation, observed on a common annual domain. Let $\mathcal{T}$ denote the annual time domain. For station $i=1,\ldots,n$ with $n=35$, let $\text{Temp}_i(s)$ denote the temperature function and $\text{Prec}_i(s)$ denote the precipitation function. Following \citet{ramsay2002}, we focus on predicting total annual precipitation from the temperature trajectory. Let
\[
y_i=\text{LogPrec}_i
\]
denote the logarithm of total annual precipitation at station $i=1,\ldots,n$, and let
\[
X_{i}(s)=\text{Temp}_i(s)
\]
represent the functional predictor; this falls into our model Eq.~\eqref{eq:FLRM} and Eq.~\eqref{eq:Centred_model} using $p=1$ functional predictors (and hence we omit the $j$ index). The scalar-on-function regression model is then given by
\begin{equation}
\text{LogPrec}_i=\alpha+\int_{0}^{\mathcal{T}}\text{Temp}_i(s)\beta(s)\,ds+\varepsilon_i.
\end{equation}
Here, $\beta(s)$ is the unknown regression coefficient function, $\alpha$ is the intercept, and $\varepsilon_i$ are independent random errors with mean zero and constant variance. Due to seasonality, we chose to use the Fourier basis with $K$ basis functions. Since $p=1$ we have $m=K$.

For clarity, we distinguish three related design matrices used throughout the empirical analysis. 
In practice, functions are observed at a finite set of $T$ time points over the annual cycle (e.g., $T=365$ for daily observations, or $T=12$ for monthly observations),  
and we denote the raw discretized temperature data by $\Wmat\in\R^{n\times T}$, whose rows correspond to stations and whose columns correspond to the temporal observations. This matrix is used to examine the dependence structure among the original predictors. 
Second, $\Zmat\in\R^{n\times m}$ denotes the functional design matrix obtained by representing each temperature curve through the chosen basis system (the Fourier basis) as given by Eq.~\eqref{eq:defZ}. Finally, $\widetilde{\Zmat}=[\mat{1}\;\;\Zmat]\in\R^{n\times(m+1)}$ denotes the augmented design matrix formed by appending a column of ones to $\Zmat$, so that the intercept $\alpha$ is estimated jointly with the remaining regression coefficients. The estimators $\b$ will be vectors of size $m+1$ where the first coordinate is the intercept $\alpha$.

\begin{figure}[ht!]
    \centering
    \includegraphics[width=0.9\linewidth]{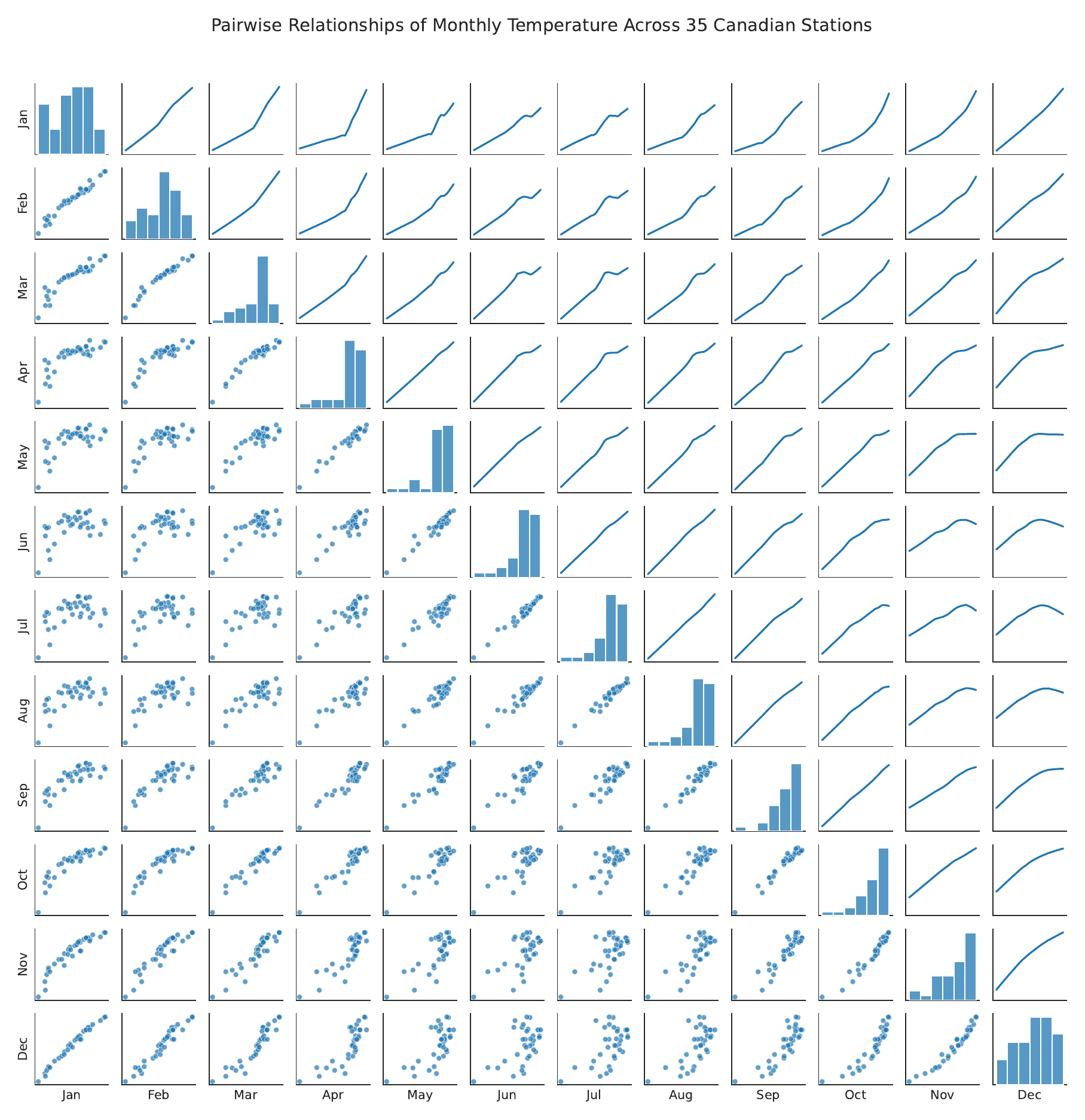}
    \caption{Monthly temperature variables in $\Wmat=(w_{ij})\in\mathbb{R}^{35\times 12}$, where $w_{ij}$ denotes the temperature recorded at station $i$ during month $j$, shown through a pairwise scatterplot matrix for the twelve months (Jan--Dec). 
    Each point represents one station observed jointly for the corresponding pair of months. Diagonal panels display the marginal distributions of each monthly variable, lower-triangular panels show pairwise scatterplots, and upper-triangular panels present smoothed trend lines between months.}
    \label{Correlation_pairplot_12x12}
\end{figure}

\subsubsection{Discrete analysis}
Before taking a FDA approach, we do an initial analysis using the raw data matrix $\Wmat$.
Each weather station reports daily data, so it is first tempting to use $T=365$ temperature covariates, but since there are only $n=35$ stations, this results in a hopelessly underdetermined regression and would give meaningless results. Thus we aggregate the daily data into monthly data and fit an OLS model to 
$T=12$ monthly temperature covariates, so $\Wmat\in\R^{n\times 12}$. The results indicate a relatively high coefficient of determination, with $R^2=0.877$ and adjusted $R^2=0.810$, suggesting a good overall fit. The overall model is statistically significant with $F=13.08$ (df $=22$) and $p<0.001$, and the residual standard error is $0.123$. However, despite the high $R^2$, several regression coefficients are not statistically significant, and the estimates exhibit instability due to strong multicollinearity among the predictors. This reflects a fundamental consequence of discretizing smooth functional data: as the number of discretization points increases, neighboring measurements become strongly dependent, which in turn leads to an ill-conditioned design matrix.

Figure~\ref{Correlation_pairplot_12x12} presents the pairwise relationships among the monthly temperature predictors, based on the matrix $\Wmat\in\R^{35\times12}$ whose rows correspond to stations and columns correspond to the twelve months. Each off-diagonal panel displays the joint observations for a pair of months with each marker showing one weather station. Strong positive associations are evident for adjacent months (e.g., January-February, June-July, and November-December), reflecting the smooth seasonal evolution of temperature over the annual cycle. More distant months also exhibit substantial dependence, often through nonlinear but monotone relationships induced by common climatic patterns. The diagonal histograms show seasonal shifts in the marginal distributions, with colder winter months and warmer summer months exhibiting distinct centers and spreads.

These observations indicate substantial dependence among the predictors arising from persistent inter-month relationships. This helps explain the instability of ordinary least squares estimates and motivates the use of regularized estimators such as ridge, functional ridge, and the proposed functional Liu method. This limitation highlights that direct discretization alone is not well suited for functional regression. Instead, dimension reduction through basis function representation and smoothness regularization is required. 

\subsubsection{FDA analysis}
In this study, we are considering a scalar-on-function linear regression model as given in Eq.~\eqref{eq:Centred_model}. Following the framework introduced in Section~\ref{sec:Methodology}, the temperature curves are represented using a Fourier basis expansion with $K=11$ basis functions, reflecting the periodic nature of the annual cycle. The basis-based design matrix $\Zmat$ defined in Eq.~\eqref{eq:defZ} and the coefficient vector defined in Eq.~\eqref{coeffvec_b}  are obtained from the Fourier basis representation of the predictor curves. For estimation with an intercept term, we use the corresponding augmented design matrix $\widetilde{\Zmat}$ introduced earlier.

For predictive assessment, the $35$ stations were randomly partitioned into a training sample of size $24$ and a test sample of size $11$. Let $\widetilde{\Zmat}_{\mathrm{train}}$ and $\widetilde{\Zmat}_{\mathrm{test}}$ denote the corresponding row-submatrices of $\widetilde{\Zmat}$ formed from the training and test observations, respectively. All model fitting and tuning were carried out using the training sample, while predictive performance was evaluated on the test sample. Using the augmented representation, these effects are further quantified through condition-number diagnostics of the design matrix. For the basis-based model, the condition number of the full-sample augmented design matrix is
\(
\kappa(\widetilde{\Zmat})=3090,
\)
while for the training design matrix it is
\(
\kappa(\widetilde{\Zmat}_{\mathrm{train}})=3387.
\)
The large condition number confirms substantial ill-conditioning and indicates the presence of multicollinearity among the predictors. 

Estimation is then carried out using the corresponding penalized functional regression formulation described in Section~\ref{sec:Methodology}, where a roughness penalty is imposed on the coefficient function $\beta(s)$ to control estimation instability and avoid overfitting.  
The resulting basis-expanded design matrix and penalty matrix that are subsequently used to compute the competing estimators  are given by
\begin{align*}
\widehat{\b}_{\text{LS}} &= \Sp\widetilde{\Zmat}^\top \y, 
&&\text{OLS}, \\
\widehat{\b}_{\text{ridge}} &= (\Smat+\lambda \Ipp)^{-1}\widetilde{\Zmat}^\top \y, 
&&\text{ridge regression},\\
\widehat{\b}_{\text{Liu}} &=
(\Smat+\Ipp)^{-1}
\left(\widetilde{\Zmat}^\top \y + d\,\widehat{\b}_{\text{LS}}\right), 
&&\text{(classical) Liu},\\
\widehat{\b}_{\cardot} &= (\Smat+\Q)^{-1}\widetilde{\Zmat}^\top \y, 
&&\text{generalized ridge with roughness},\\
\widehat{\b}_{\text{fLiu}} &= (\Smat+\Q)^{-1}
\left(\widetilde{\Zmat}^\top \y+d\,\Qmat\widehat{\b}_{\text{LS}}\right), 
&&\text{proposed fLiu}
\end{align*}
where $\Smat=\widetilde{\Zmat}^\top\widetilde{\Zmat}\in\R^{m+1 \times m+1}$, $\Rmat\in\R^{m\times m}$ is the roughness penalty matrix defined in subsection \ref{Basis and roughness penality} and $\Q=\lambda\bigl(\alpha \Ipp+(1-\alpha)\Rmat_0\bigr)$. 
Here, $\Rmat_0\in\R^{(m+1) \times (m+1)}$ is $\Rmat$ augmented with a leading row and column of zeros to account for the intercept (which we do not want to penalize), $\Rmat_0=
\begin{pmatrix}
0 & \mathbf{0}^\top\\
\mathbf{0} & \Rmat
\end{pmatrix}$. 
For tuning parameters, the smoothing parameters $\lambda$ and $\alpha$ and Liu-type biasing parameter $d$ are selected using the GCV criterion as implemented in the package \texttt{fLiu.py} mentioned earlier.  The parameter $\alpha$ was restricted to $[0,1]$, the parameter $\lambda$ was restricted to $[10^{-6},10^6]$ (and was log-transformed for the optimization), and $d$ was restricted to $[-1000,1]$.
The resulting GCV values and corresponding tuning parameters for each estimator are reported in Table~\ref{tab:model_comparison_gcv}.
\begin{table}[htbp]
	\centering
	\caption{GCV values and corresponding tuning parameters}
	\label{tab:model_comparison_gcv}

    \begin{tabular}{lccccc}
		\toprule
		Estimator & $\lambda$ & $d$ & $\alpha$ & $\mathrm{GCV}$ \\
		\midrule
		OLS    & --        & --      & --      & 0.0264 \\
		Ridge  & 0.001000   & --      & --      & 0.0264 \\
		Liu    & 0.000001   & $-749.75$ & --      & 0.0264 \\
        Generalized Ridge & 2.5526647   & --      & 0.0004 & 0.0202 \\
		\textbf{fLiu}   & 0.007900   & $-332.6659$ & 0.0029 & \textbf{0.0133} \\
		\bottomrule
	\end{tabular}
	
\end{table} 

Table~\ref{tab:model_comparison_gcv} reports the optimal tuning parameters together with the corresponding GCV values for the Canadian weather data across all estimators. Since GCV is primarily useful for relative comparison, the results indicate similar predictive performance for OLS, ridge, Liu, and generalized ridge, all of which yield comparable values. The proposed \texttt{fLiu} estimator attains the smallest GCV value among all competing methods, indicating the best relative predictive performance under the selected tuning parameters. 
\begin{figure}[ht]
    \centering
    \includegraphics[width=0.9\linewidth]{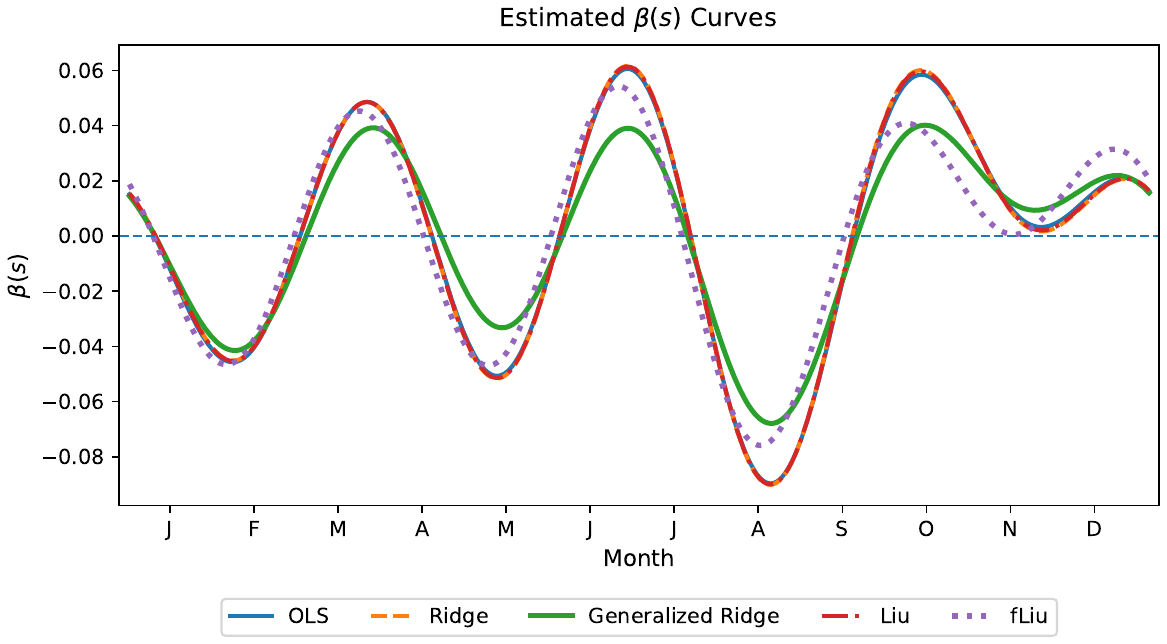}
    \caption{Estimated coefficient functions $\beta(s)$ for OLS, Ridge, Liu, Generalized Ridge, and \texttt{fLiu} over the annual cycle. All estimators exhibit similar seasonal patterns, while the Generalized Ridge estimator shows the largest deviation.}
    \label{fig:beta_2D_plot}
\end{figure}

Figure~\ref{fig:beta_2D_plot} displays the estimated coefficient functions for the monthly temperature predictor across the competing estimators. All methods recover a similar seasonal pattern, indicating a stable relationship between annual temperature trajectories and log annual precipitation. The curves alternate between positive and negative regions, suggesting varying seasonal effects on precipitation. Differences among estimators are modest overall, with \texttt{fLiu} showing stronger shrinkage by remaining closer to zero over much of the cycle, while the Generalized Ridge estimator also produces a comparatively smooth curve.
These results suggest that combining shrinkage and smoothness regularization can provide additional predictive gains over the standard alternatives. In particular, the superior GCV performance of \texttt{fLiu} suggests that the additional flexibility provided by its tuning parameters is beneficial.

\subsection{Results in the overdetermined case}
A central question of this empirical study is whether the proposed functional Liu-type estimator (\texttt{fLiu}) can improve predictive performance relative to conventional penalized approaches in scalar-on-function regression.
\begin{table}[ht]
	\centering
	\caption{Prediction accuracy (training and testing loss) for each estimator}
	\label{tab:prediction_accuracy}

    \begin{tabular}{lcc}
		\toprule
		Estimator & $\mathrm{Training\ loss}$ & $\mathrm{Testing\ loss}$ \\
		\midrule
		OLS    & 0.0066 & 0.0389 \\
		Ridge  & 0.0066 & 0.0374 \\
		Liu    & 0.0066 & 0.0378 \\
        Generalized Ridge & 0.0073 & 0.0371 \\
		\textbf{fLiu} & 0.0131 & \textbf{0.0203} \\
		\bottomrule
	\end{tabular}
\end{table}
Prediction accuracy across estimators is summarized in Table~\ref{tab:prediction_accuracy}. In predicting total annual precipitation from annual temperature variation, the estimators display different training--testing trade-offs. OLS attains low training loss but the highest testing loss, indicating overfitting without regularization. Ridge, generalized ridge, and Liu provide modest improvements over OLS, while the proposed \texttt{fLiu} estimator achieves the lowest testing loss. Although \texttt{fLiu} has a higher training loss, it delivers superior out-of-sample performance through combined shrinkage and smoothness regularization.

\begin{figure}[ht!]
\centering
\includegraphics[width=0.67\textwidth]{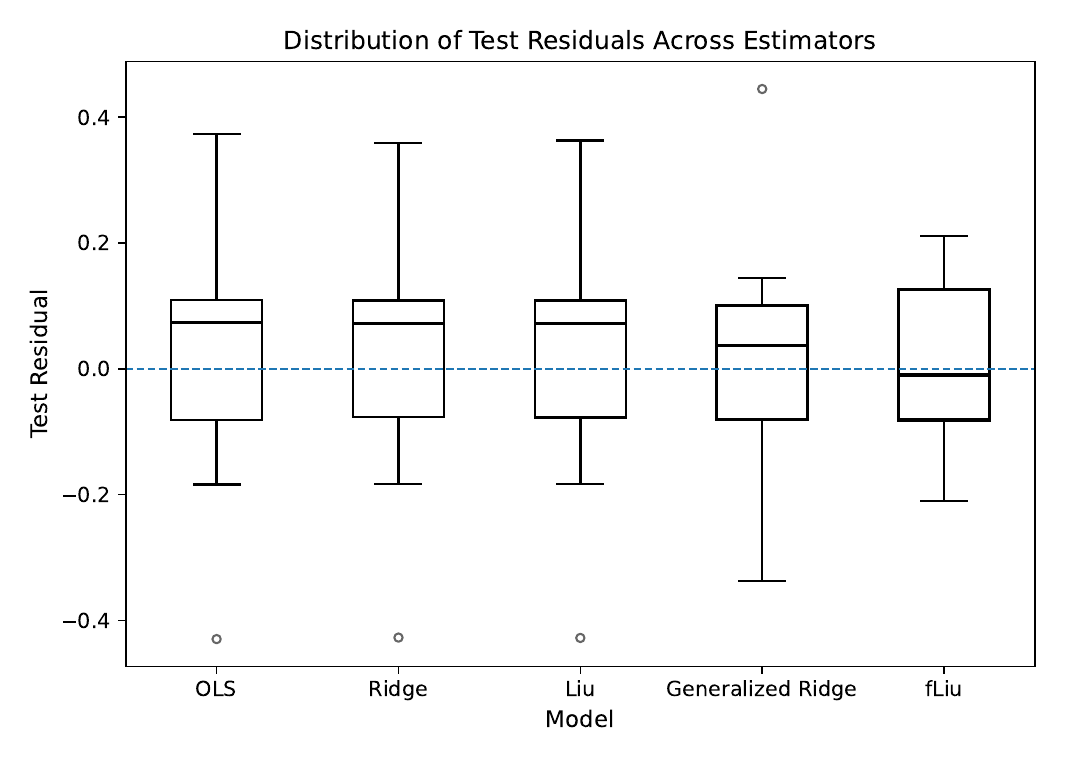}
\caption{Distribution of test residuals across the five competing estimators. 
The box plot summarizes the median, dispersion, and potential outliers of prediction errors on the test set. 
}
\label{fig:boxplot_residuals}
\end{figure}

Figure~\ref{fig:boxplot_residuals} displays the distribution of test residuals for the five competing estimators. 
While the numerical test-loss values summarize average prediction error, the box plot provides additional insight into the variability of the residuals, since estimators with lower variability give more stable performance.
 The fLiu estimator exhibits a comparatively tighter residual distribution, which supports the superior out-of-sample performance observed in the test-loss comparison.

From an application perspective, these results illustrate the effectiveness of penalized functional regression models in capturing the smooth temporal structure of temperature data. The proposed \texttt{fLiu} estimator demonstrates stable and interpretable behavior, while maintaining predictive performance comparable to existing approaches. Overall, it provides a flexible and well-balanced regularization strategy within the class of penalized functional regression methods.

\subsection{Underdetermined setting and degeneracy of GCV}
To illustrate the degeneracy result, we construct an underdetermined version of the Canadian weather model by increasing the number of basis functions so that the number of coefficients exceeds the sample size. In particular, we use a richer basis expansion so that $m = pK > n$ with $K=35$, placing the model in the high-dimensional (underdetermined) regime. In this setting, the GCV criterion is observed to be  constant with respect to $d$, indicating that it provides no guidance for selecting the Liu parameter. This behavior is consistent with Theorem~\ref{thm:degeneracy of gcv}, which shows that GCV becomes independent of $d$ in the underdetermined case. To address this issue, we use the proposed plug-in rule to obtain a data-driven estimate of $d$. This provides a well-defined choice of $d$ even when cross-validation criteria fail.
\begin{table}[ht]
\centering
\caption{Comparison of GCV-based and plug-in tuning for the fliu parameter}
\label{tab:degeneracy_weather}

\begin{tabular}{l c c c c}
\toprule
\multicolumn{5}{c}{\textbf{Overdetermined setting ($K=11$)}} \\
\midrule
Method & $\lambda$ & $\alpha$ & $d$ & Test loss \\
\midrule
GCV-based fLiu         & 0.0079 & 0.0029 & $-332.6660$ & 0.0203 \\
Plug-in tuned fLiu     & 0.0079 & 0.0029 & $-54.9469$  & 0.0339 \\
Projected plug-in fLiu & 0.0079 & 0.0029 & 0.0000    & 0.0388 \\
\midrule
\multicolumn{5}{c}{\textbf{Underdetermined setting ($K=35$)}} \\
\midrule
Method & $\lambda$ & $\alpha$ & $d$ & Test loss \\
\midrule
GCV-based fLiu         & 2.6781 & 0.0003 & 0.0000   & 0.0382 \\
Plug-in tuned fLiu     & 2.6781 & 0.0003 & $-0.0848$  & 0.0532 \\
Projected plug-in fLiu & 2.6781 & 0.0003 & 0.0000   & 0.0382 \\
\bottomrule
\end{tabular}
\end{table}

Table~\ref{tab:degeneracy_weather} shows that in the overdetermined setting both tuning approaches perform competitively, whereas in the underdetermined regime the plug-in rule substantially improves prediction accuracy relative to the GCV-based choice (where we arbitrarily chose $d=0$). The projected plug-in estimator, which restricts $d$ to the interval $[0,1]$, provides a useful benchmark and in some cases coincides with the GCV-based solution.

Figure~\ref{fig:gcv_parallel_both cases} highlights the sharp contrast between the two dimensional regimes. In the overdetermined setting ($K=11$), the GCV criterion changes noticeably across the range of $d$, indicating that cross-validation remains informative for selecting the Liu parameter. In contrast, when the basis dimension is increased to $K=35$, the problem becomes underdetermined and the GCV curve is essentially flat. Hence, the criterion provides no meaningful guidance for choosing $d$, which is consistent with the theoretical degeneracy established in Theorem~\ref{thm:degeneracy of gcv}. This illustrates the practical need for alternative tuning strategies, such as the proposed plug-in rule, in high-dimensional functional regression problems.
\begin{figure}[ht]
\centering

    \includegraphics[width=.48\textwidth]{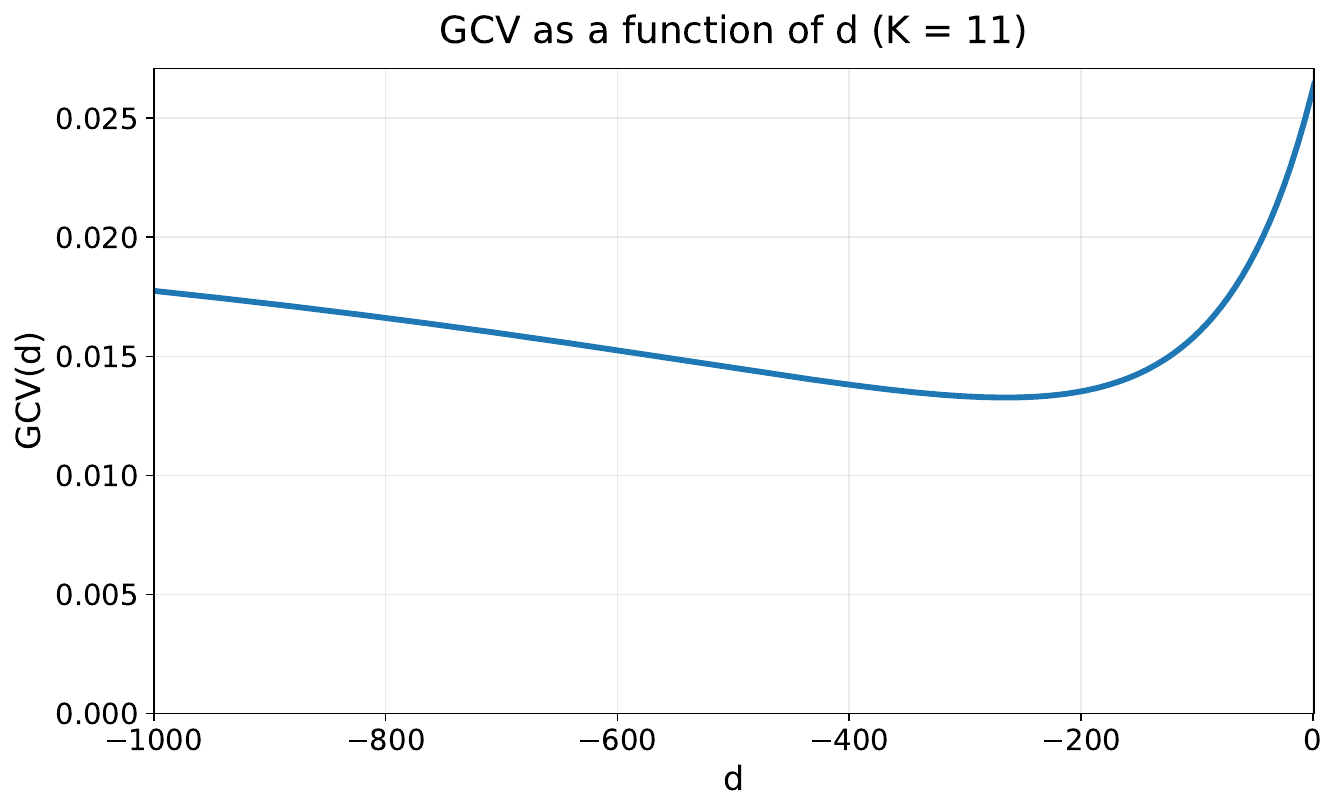}
    \centering
    \includegraphics[width=.48\textwidth]{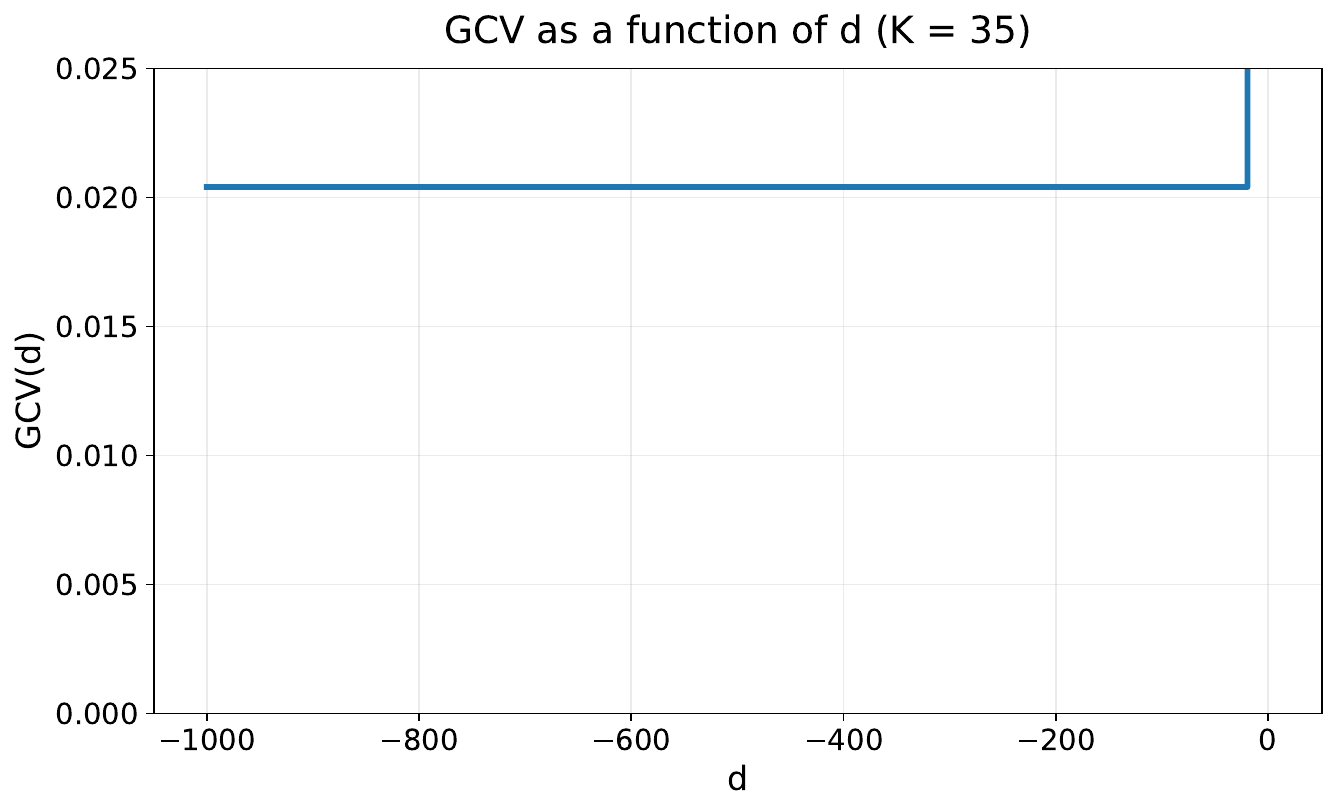}

\caption{GCV as a function of the Liu parameter $d$ for the Canadian Weather data. Left panel: over-determined setting ($K=11$), where the GCV criterion varies substantially with $d$. Right panel: under-determined setting ($K=35$), where the GCV curve is essentially constant, confirming the degeneracy result of Theorem~\ref{thm:degeneracy of gcv}.}
\label{fig:gcv_parallel_both cases}
\end{figure}

\section{Conclusion}

This paper proposes and analyzes a functional Liu-type shrinkage estimator (\texttt{fLiu}) for scalar-on-function regression, motivated by the instability induced by multicollinearity among discretized or basis-expanded functional predictors. The proposed methodology combines directional shrinkage with smoothness regularization, thereby extending classical Liu-type ideas to the functional regression setting. In addition to introducing the estimator, the study established theoretical and empirical results that clarify its behavior across different dimensional regimes.

From a theoretical perspective, the study derived an explicit mean squared error decomposition for the proposed estimator, clarifying the separate effects of variance reduction and shrinkage bias. The resulting risk was shown to be a convex quadratic function of the Liu parameter $d$, yielding a unique optimal choice and providing the basis for the proposed plug-in tuning rule.

A further contribution of the paper concerns parameter selection in functional Liu-type regression. A key challenge in developing Liu methodology for functional models is the choice of tuning parameters, particularly in high-dimensional (underdetermined) settings where the model dimension exceeds the sample size. We showed that, in such regimes, generalized cross-validation (GCV) and leave-one-out cross-validation (LOO-CV) become constant with respect to the Liu shrinkage parameter $d$, and therefore fail to provide useful guidance for its selection. By contrast, in the overdetermined setting these criteria vary meaningfully with $d$, allowing effective data-driven tuning. This contrast provides theoretical insight into the behavior of standard selection criteria and motivates the use of alternative strategies, such as the proposed plug-in rule, in high-dimensional applications.

Experimentally, first, we evaluated the performance of the proposed \texttt{fLiu} estimator relative to OLS, ridge, classical Liu, and functional generalized ridge estimators under a GCV-based tuning framework. The results showed that \texttt{fLiu} produces stable and interpretable coefficient estimates while maintaining strong predictive performance. In particular, the estimator provides a flexible compromise between shrinkage and smoothness, performing competitively with existing penalized approaches while remaining robust under substantial predictor dependence.

Second, we illustrated the practical utility of the proposed method through prediction of total annual precipitation at Canadian weather stations using annual temperature trajectories. The empirical results show that seasonal temperature profiles contain substantial predictive information for precipitation outcomes, and that the proposed estimator performs effectively in this structured functional setting.

Third, we examined the dependence structure among stations through the seasonal behavior of the functional predictors. The pairwise relationships among monthly temperature variables showed strong correlation induced by common annual cycles, leading to pronounced multicollinearity in the associated design matrix. This observation reinforces the importance of functional representations together with regularization when modeling correlated climatic processes.

Overall, the findings demonstrate that the proposed \texttt{fLiu} estimator constitutes a stable, flexible, and practically effective addition to the class of penalized functional regression methods. By integrating shrinkage and smoothing within a unified framework, it offers a principled approach for handling highly correlated functional predictors while preserving interpretability and predictive reliability.

Several directions remain for future research. Extending the framework to function-on-function regression and generalized functional models would broaden its applicability. 
Another interesting direction would be theoretical analysis of the properties of our plug-in estimate.
Finally, incorporating sparsity into the functional Liu-type framework could enhance interpretability when multiple correlated functional predictors are present.

\section*{Acknowledgments}
ChatGPT (version GPT-5.3) was used for Language improvement, Interactive online search with LLM-enhanced search engines, Literature classification and 
Coding assistance.

\section*{Data Availability Statement}
The data that support the findings of this study are openly available 
\ifisanonymous
(link redacted during anonymous review).
\else
in the \texttt{fliu} repository at \url{https://github.com/stephenbeckr/functional-Liu}.
\fi

\bibliography{references}

@article{barata2012moore,
  title={The {Moore-Penrose} pseudoinverse: A tutorial review of the theory},
  author={Barata, Jo{\~a}o Carlos Alves and Hussein, Mahir Saleh},
  journal={Brazilian Journal of Physics},
  volume={42},
  number={1},
  pages={146--165},
  year={2012},
  publisher={Springer}
}

@article{akdeniz2007,
  author    = {Akdeniz, F. and Ka{\c{c}}iranlar, S.},
  title     = {A generalized {L}iu estimator for handling multicollinearity},
  journal   = {Journal of Statistical Planning and Inference},
  volume    = {137},
  pages     = {1872--1880},
  year      = {2007}
}

@article{cardot2003,
  author    = {Cardot, H. and Ferraty, F. and Sarda, P.},
  title     = {Spline estimators for the functional linear model},
  journal   = {Statistica Sinica},
  volume    = {13},
  number    = {3},
  pages     = {571--591},
  year      = {2003}
}

@article{craven1979,
  author    = {Craven, P. and Wahba, G.},
  title     = {Smoothing noisy data with spline functions: Estimating the correct degree of smoothing by the method of generalized cross-validation},
  journal   = {Numerische Mathematik},
  volume    = {31},
  number    = {4},
  pages     = {377--403},
  year      = {1979}
}

@article{eilers1996,
  author    = {Eilers, P. H. and Marx, B. D.},
  title     = {Flexible smoothing with {B}-splines and penalties},
  journal   = {Statistical Science},
  volume    = {11},
  number    = {2},
  pages     = {89--121},
  year      = {1996}
}

@article{hoerl1970,
  author    = {Hoerl, A. E. and Kennard, R. W.},
  title     = {Ridge regression: Biased estimation for nonorthogonal problems},
  journal   = {Technometrics},
  volume    = {12},
  number    = {1},
  pages     = {55--67},
  year      = {1970}
}

@book{tikhonov1977solutions,
  author    = {Tikhonov, Andrey N. and Arsenin, Vasiliy Y.},
  title     = {Solutions of Ill-Posed Problems},
  publisher = {Winston and Sons},
  address   = {Washington, DC},
  year      = {1977}
}

@article{kibria2003,
  author    = {Kibria, B. M. G.},
  title     = {Some improved ridge regression estimators and their applications},
  journal   = {Journal of Modern Applied Statistical Methods},
  volume    = {2},
  pages     = {133--144},
  year      = {2003}
}

@article{liu1993,
  author    = {Liu, K.},
  title     = {A new class of biased estimate in linear regression},
  journal   = {Communications in Statistics --- Theory and Methods},
  volume    = {22},
  pages     = {393--402},
  year      = {1993}
}

@article{liu2003,
  author    = {Liu, K.},
  title     = {On the statistical properties of the {L}iu estimator},
  journal   = {Communications in Statistics --- Theory and Methods},
  volume    = {32},
  number    = {5},
  pages     = {1009--1020},
  year      = {2003}
}

@article{mehrotra2022,
  title={Simultaneous variable selection, clustering, and smoothing in function-on-scalar regression},
  author={Mehrotra, Suchit and Maity, Arnab},
  journal={Canadian Journal of Statistics},
  volume={50},
  number={1},
  pages={180--199},
  year={2022},
  publisher={Wiley Online Library}
}

@book{ramsay2002,
  author    = {Ramsay, J. O. and Silverman, B. W.},
  title     = {Applied Functional Data Analysis: Methods and Case Studies},
  publisher = {Springer},
  address   = {New York},
  edition   = {1st},
  year      = {2002}
}

@book{ramsay2005,
  author    = {Ramsay, J. O. and Silverman, B. W.},
  title     = {Functional Data Analysis},
  edition   = {2},
  publisher = {Springer},
  year      = {2005}
}

@book{RamsayHookerGraves2009,
  author    = {James O. Ramsay and Giles Hooker and Spencer Graves},
  title     = {Functional Data Analysis with R and MATLAB},
  publisher = {Springer},
  address   = {New York},
  year      = {2009}
}

@article{gruber2010,
  author  = {Gruber, Marvin H. J.},
  title   = {Liu and Ridge Estimators: A Comparison},
  journal = {Communications in Statistics---Theory and Methods},
  year    = {2010},
  volume  = {39},
  number  = {8},
  pages   = {1485--1494},
  doi     = {10.1080/03610920902873917}
}

@article{filzmoser2016,
  author  = {Filzmoser, Peter and Kurnaz, Fatih Serkan},
  title   = {A Robust Liu Regression Estimator},
  journal = {Communications in Statistics---Theory and Methods},
  year    = {2016},
  volume  = {47},
  number  = {6},
  pages   = {1285--1298},
  doi     = {10.1080/03610926.2016.1271889}
}

@article{ozkale2007,
  author  = {\"Ozkale, M. R. and Ka{\c{c}}iranlar, S.},
  title   = {The restricted and unrestricted two-parameter estimator},
  journal = {Statistics and Probability Letters},
  year    = {2007},
  volume  = {77},
  number  = {4},
  pages   = {438--446},
  doi     = {10.1016/j.spl.2006.06.018}
}

\end{document}